\documentclass[lettersize,journal]{IEEEtran}
\usepackage{amsmath,amsfonts}
\usepackage{algorithmic}
\usepackage{algorithm}
\usepackage{array}
\usepackage{textcomp}
\usepackage{stfloats}
\usepackage{url}
\usepackage{verbatim}
\usepackage{cite}
\usepackage{sharina}
\usepackage{setspace}
\usepackage{amsthm}
\usepackage{amssymb,amsmath,latexsym}
\usepackage{graphicx,subfigure}
\usepackage{color,cite,times}
\usepackage{epstopdf}
\usepackage{multirow}
\usepackage{array}
\usepackage{booktabs}
\usepackage{hhline} 
\newtheorem{assumption}{Assumption}

\newtheorem{theorem}{\textbf{Theorem}}

\allowdisplaybreaks[4]

\begin{document}
	
	\title{
	Learning Performance Optimization for Edge AI System with Time and Energy Constraints}
	\author{	Zhiyuan~Zhai, Wei Ni, \IEEEmembership{Fellow, IEEE},
	and Xin~Wang,~\IEEEmembership{Fellow, IEEE}
}

	\maketitle

\begin{abstract}
	Edge AI, which brings artificial intelligence to the edge of the network for real-time processing and decision-making, has emerged as a transformative technology across various applications. However, the deployment of Edge AI systems faces significant challenges due to high energy consumption and extended operation time.
	In this paper, we consider  an Edge AI system which integrates the  data acquisition,  computation and communication processes, and  focus on improving learning performance of this system.
We  model the time and energy consumption of different processes and  perform a rigorous convergence analysis to quantify the impact of key system parameters, such as the amount of collected data and the number of training rounds, on the learning performance. Based on this analysis, we formulate a system-wide optimization problem that seeks to maximize learning performance under given time and energy constraints. We explore both homogeneous and heterogeneous device scenarios, developing low-complexity algorithms based on one-dimensional search and alternating optimization to jointly optimize data collection time and training rounds. Simulation results validate the accuracy of our convergence analysis and demonstrate the effectiveness of the proposed algorithms, providing valuable insights into designing energy-efficient Edge AI systems under real-world conditions.

\end{abstract}

\begin{IEEEkeywords}
Edge AI, energy efficiency,  convergence analysis,   alternating optimization.
\end{IEEEkeywords}

\section{Introduction}
{Edge Artificial Intelligence (AI) is emerging as a key technology that brings intelligence directly to the network's edge, enabling real-time data processing and decision-making close to the source of data~\cite{li2023high,zhang2018enabling,sonmez2020machine}. This paradigm shifts computation and model inference tasks from centralized cloud servers to edge devices, offering several advantages, including reduced latency, enhanced data privacy, and minimized bandwidth usage~\cite{liu2022resource,liang2020ai}. Edge AI is particularly beneficial in applications, such as autonomous vehicles, smart cities, and healthcare, where immediate data processing and decision-making are crucial~\cite{hong2021ai,zhai2025distributed}. By leveraging local computation resources at the edge, Edge AI enables faster responses and more context-aware services, significantly improving system performance and user experience~\cite{xu2021edge}. 

\subsection{Motivation and Challenges}
Despite the above advantages, deploying Edge AI in practical environments faces significant challenges, particularly related to high energy consumption and prolonged operation time \cite{alom2019state,vestias2020moving,zhai2024over,zhai2025spectral}. Edge devices, such as sensors and cameras, often need to process large amounts of data continuously, leading to substantial power consumption \cite{lavin2016fast}. Limited battery life exacerbates this issue, especially in scenarios where recharging is difficult \cite{lavin2016fast}. Furthermore, many Edge AI tasks, such as federated learning, require multiple rounds of model training, which increases both time and energy demands \cite{howard2017mobilenets}. The problem becomes even more pronounced when Edge AI systems operate in remote or harsh environments, such as rural areas or offshore installations, where resources are scarce and costly to replenish \cite{li2020secured}. These challenges make it difficult to deploy Edge AI systems efficiently in environments constrained by both energy and time.
In real-world edge systems, energy resources are often constrained, and the time allocated for task execution is often limited \cite{zhai2022energy}. Thus, designing Edge AI systems with energy efficiency as a primary focus is critical for ensuring practical and sustainable deployment. 
Given the increasing reliance on edge devices for real-time AI-powered applications, energy-efficient system design is essential for maintaining system performance, extending device lifetime, and supporting continuous operation in dynamic environments. Consequently, energy efficiency has become a vital consideration in Edge AI, driving the need for novel approaches to optimize resource allocation and reduce the system's overall energy footprint \cite{mao2024green}.

\subsection{Related Work}
In response to the pressing need for energy-efficient Edge AI systems, numerous researchers have proposed various approaches to address the challenges of high energy consumption and extended operation time. For example,  \cite{chen2018thriftyedge} proposes joint resource allocation strategies in wireless edge AI, balancing energy efficiency and model accuracy by dynamically adjusting wireless transmission power and computational resources to extend device battery life while enabling efficient AI model training. \cite{li2019edge} focuses on reducing the number of communication rounds required, lowering the energy cost of running federated learning models on edge devices. Furthermore, \cite{deng2020edge} introduces an edge intelligence framework that integrates AI with wireless communication protocols to reduce energy consumption by leveraging 5G networks to enable low-latency AI tasks through efficient data and computation offloading. Moreover, studies \cite{zhang2022scalable,fang2022ai,wang2021towards} have contributed to advancing energy-efficient Edge AI systems, further driving the development of sustainable Edge AI solutions.

Although these works \cite{chen2018thriftyedge,li2019edge,deng2020edge,zhang2022scalable,fang2022ai,wang2021towards} aim to improve the energy efficiency of Edge AI systems, they primarily focus on heuristic approaches that optimize individual procedure of the overall system.  However,  an Edge AI system encompasses multiple processes, including data collection, local computation, and communication. Such coarse-grained methods lack the analytical guarantees necessary to ensure optimal system performance. Furthermore, local improvements may not lead to optimal system-wide outcomes. To fully realize energy-efficient Edge AI systems, there is a pressing need for performance analysis that considers the entire system processes. Such holistic approach would enable more robust optimization strategies that comprehensively enhance energy efficiency and overall system performance.

\subsection{Contributions}
In this paper, we propose a novel approach to enhance overall Edge AI system performance under time and energy constraints. Specifically, we develop a general Edge AI system framework that encompasses the entire process from data collection and local computation to communication during deployment and operation. We conduct a rigorous convergence analysis to quantify the impact of key system parameters, such as the amount of data collected and the number of training rounds, on system learning performance. Based on this analysis, we formulate a new problem to maximize the learning performance while adhering to time and energy limitations.
We explore this optimization problem under both homogeneous and heterogeneous device scenarios, developing low-complexity algorithms based on one-dimensional search and alternating optimization to jointly optimize data collection time and the number of training rounds.
 We summarize our contributions as follows.
\begin{itemize}
	\item  We propose an Edge AI system framework that integrates the entire process of  data collection,  local model computation, and communication. Within this framework, we  model the time and energy consumption  of the system, providing a detailed  representation of the resource requirements throughout the deployment (data collection) and operational phases (model computation and communication).
	\item  Under the proposed framework, we conduct rigorous analysis to  characterize the system's  convergence performance. This analysis quantitatively captures the impact of the amount of data collected and the number of training rounds on the learning performance. Building on this convergence analysis, we formulate a time resource allocation problem to maximize system performance within given time and energy constraints.
	\item We address  the optimization problem in both homogeneous and heterogeneous device scenarios. For the homogeneous case, we exploit the problem’s monotonic structure to reformulate it into a single-variable form solvable via one-dimensional search with global optimality. For the heterogeneous case, we reveal the problem’s block-wise convexity and apply alternating optimization to iteratively refine sensing time and training rounds. Both methods ensure monotonic convergence and low computational complexity.
\end{itemize}

The simulation results confirm the effectiveness of the proposed scheme and provide valuable insights into the behavior of different  scenarios in Edge AI system. Specifically, the numerical results  validate the accuracy of the derived convergence bound. We also obtain valuable design insights  by evaluating the homogeneous and heterogeneous device scenarios.
 Moreover, the comparison with benchmark methods demonstrate that our approach achieves significant performance improvements and attains near-optimal learning performance with limited resources.

The remainder of this paper is organized as follows. In Section II, we provide a detailed description of the data collection, local computation, and communication models for the Edge AI system. In Section III, we analyze the system’s learning performance and formulate the corresponding optimization problem. Section IV discusses and solves the optimization problem under both homogeneous and heterogeneous  scenarios. Section V presents the simulation results. Finally the paper concludes with remarks in Section VI.

\textit{Notation:} We denote sets using calligraphic symbols such as $\mathcal{D}$ for the dataset, and $\mathbb{R}$ and $\mathbb{C}$ for the sets of real and complex numbers, respectively. Scalars are represented by lowercase regular letters, vectors by lowercase bold letters, and matrices by bold capital letters.  $(\cdot)^\ast$, $(\cdot)^\mathrm{T}$, $(\cdot)^\mathrm{H}$, and $(\cdot)^{\dagger}$ are used to denote the conjugate, transpose, conjugate transpose, and pseudoinverse, respectively. The $i$-th entry of a vector $\mathbf{x}$ is denoted by $x_i$, and the $(i,j)$-th entry of a matrix $\mathbf{X}$ is $X_{ij}$. The $l_2$-norm and Frobenius norm are denoted by $|\cdot|_2$ and $|\cdot|_F$, respectively. The expectation operator is represented by $\mathbb{E}[\cdot]$, and the trace of a matrix is denoted by $\mathrm{Tr}(\cdot)$.

	\section{System Model}\label{sysmodel}
Edge AI systems, e.g., federated learning~\cite{amiri2020federated,chen2020joint,qin2021federated,yang2020energy}, have attracted significant attention in recent years due to their potential in distributed intelligent computing. Given the constrained resources at the network edge, it is crucial to design energy-efficient Edge AI systems to maximize learning performance. Existing studies mainly focus on optimizing the operational phase---local model computation and global parameter communication---while often overlooking the critical data acquisition process, resulting in suboptimal overall performance. 

In practice, an Edge AI system involves a complete workflow that spans from data acquisition and preprocessing to local computation and parameter transmission. These stages are inherently interdependent, and optimizing any individual stage in isolation may not yield system-level optimality. Therefore, a holistic framework that jointly models and optimizes data collection, computation, and communication is essential for enhancing the global learning performance across diverse application scenarios. We introduce this unified framework in the following sections.


\subsection{Edge AI Learning  Model}
The goal of  Edge AI system is to minimize an ideal loss function \cite{mohri2019agnostic}:
\begin{align}\label{eq01}
	\min_{\xx \in \mathbb{R}^{D \times 1}} F(\xx) = \frac{1}{M_{\text{total}}} \sum_{m=1}^{M_{\text{total}}} f(\xx; \xi_m),
\end{align}
where $\xx$ represents the $d$-dimensional model parameter vector, $M_{\text{total}}$ is the total number of training samples in the environment, $\xi_m$ is the $m$-th training sample, and $f(\xx; \xi_m)$ is the loss function corresponding to this sample. 

The training data in the environment is collected by the edge devices, with the $k$-th device holding $M_k$ samples and $\sum_{k=1}^{K}M_k=M$. The local dataset of the $k$-th device is denoted as $\mathcal{D}_m = \{\xi_{km}: 1 \leq m \leq M_k\}$. Typically, the  edge devices cannot collect all data from environment, we have $M\leq M_{\text{total}}$. Then the practical  learning task can  be reformulated as
\begin{align}\label{eq02}
	&\tilde{F}(\xx) = \frac{1}{M} \sum_{k=1}^{K} M_k F_k(\xx; \mathcal{D}_k),\\\text{with}  \quad& F_k(\xx; \mathcal{D}_k) = \frac{1}{M_k} \sum_{\xi_{km} \in \mathcal{D}_k} f(\xx; \xi_{km}).
\end{align}

After data acquisition, the edge devices 
update their local models by minimizing their respective local objective functions $F_k$, while the server aggregates the global model based on the updates.  The model $\xx$ is iteratively refined over $I$ rounds of training. For each round $i$ ($1 \leq i \leq I$), the learning procedure involves the following steps:

\begin{itemize}
	
	\item \textit{Model broadcast}: The central server sends the current global model $\xx_i$ to all edge devices.
	\item \textit{Local gradient computation}: Each  device calculates the local gradient based on its dataset. The gradient is expressed as:
	\begin{align}\label{eq03}
		\gg_{k,i} = \nabla F_k(\xx_i; \mathcal{D}_k) \in \mathbb{R}^D,
	\end{align}
	where $\nabla F_k(\xx_i; \mathcal{D}_k)$ is the gradient of $F_k(\cdot)$ evaluated at $\xx = \xx_i$.
	
	\item \textit{Model update}: The edge  devices transmit their local gradients $\gg_{k,i}$  to the central server via wireless communication channels. The server aggregates these gradients to form a global update, denoted as $\gg_i = \frac{1}{M}\sum_{k=1}^{K} M_k \gg_{k,i}$.  The global model is updated using the following rule:
	\begin{align}\label{eq04}
		\xx_{i+1} = \xx_i - {\lambda}{\gg}_i,
	\end{align}
	where   $\lambda$ represents the learning rate.
\end{itemize}

We show the workflow in Fig.~\ref{time_schedule} and summarize the overall framework of  Edge AI system  in Algorithm \ref{alg:FL_framework}.

\begin{figure}[h]
	\centering            
	\includegraphics[width=1\linewidth]{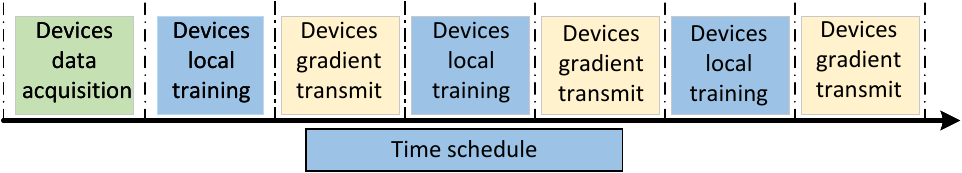}
	\caption{Workflow of Edge AI system.}
	\label{time_schedule}
\end{figure}
\begin{algorithm}[htb]
	\caption{Overall Framework for Edge AI System} 
	\label{alg:FL_framework} 
	\begin{algorithmic}[1] 
		\REQUIRE Training round $I$, data sensing times $\{t_{\text{sens},k}\}_{k=1}^{K}$. 
		\STATE {\textbf{Initialization:}} Edge devices collect data from environment with time duration  $\{t_{\text{sens},k}\}_{k=1}^{K}$.
		\FOR{ $i \in [I]$ }
		\STATE Central server broadcasts the global model $\xx_i$ to edge devices.
		\FOR{$k \in [K]$ in parallel}
		\STATE  Device $k$ computes its local gradient  $ \gg_{k,i}$  using \eqref{eq03};
		\STATE  Device $k$ transmits its local gradient $\gg_{k,i}$ to the central server via wireless communication.
		\ENDFOR
		\STATE Central server aggregates the received gradients from  edge devices to update the  global model using \eqref{eq04}.
		\ENDFOR 
	\end{algorithmic}
\end{algorithm}

\subsection{Data Acquisition, Computation and Communication Model}
Consider an  Edge AI system consisting of a central server and $K$ edge devices.  The  objective of the system is to learn a common  model by collaboration between devices and the central server.
Since raw data is generated in the environment, devices must collect the data to enable local model training. 
Each edge device is equipped with sensors for data collection, computational resources for local model training, and an antenna for communication with the server. The central server coordinates all edge devices over wireless network to perform cooperative model training. Hereinafter, the  model trained on each edge device side is referred to as the local  model, whereas the  model obtained by the central server  through collecting and aggregating all edge devices' local models is referred to as the global  model. 
The operation of the Edge AI system involves three key tasks \cite{mao2024green,shi2020communication}: data acquisition, local computation, and communication. Each task has associated energy consumption and time costs, as described in the following. 

\subsubsection{Data Acquisition}
{
Since raw data are generated in the environment, they need to be collected by edge devices for model training. This task is crucial in scenarios where edge devices have both sensing and computation capabilities. Typical applications include training models for learning activity prediction using physiological data collected by smartwatches and continuously updating AI models for Internet of Things (IoT) applications with newly available local data. Notably, in Edge AI systems, the learning model performance significantly depends on the quality and quantity of training data processed locally~\cite{liu2020data}, which are restricted by the limited battery energy of edge devices.

In Edge AI, data acquisition takes place at the initial stage and is performed only once. In the sensing process, each edge device collects surrounding environmental data using its sensors. The energy consumption during sensing operations is mainly influenced by the sensing mechanism, sampling rate, and sample quality. These factors directly affect the overall power consumption. 

Suppose there are a total of $M_{\text{total}}$ data samples distributed in the environment. A set of $K$ edge devices collaboratively perform data sensing, where device $k$ actually collects $M_k$ samples during the sensing phase. Due to practical constraints such as limited energy and time, the total number of collected samples generally does not exceed the available data in the environment, i.e., $\sum_{k=1}^{K} M_k \leq M_{\text{total}}$.

Let $t_{\text{sens},k}$ be the sensing time of edge device $k$, the volume of samples that edge device $k$  collects within $t_{\text{sens},k}$ is given by
\begin{equation}
	M_k = t_{\text{sens},k} f_{\text{s},k},
\end{equation}
where $f_{\text{s},k}$ is the sampling rate of edge device $k$. 

To reflect the diversity in sensing mechanisms and hardware characteristics, we adopt a generalized energy consumption model:
\begin{equation}
	E_{\text{sens},k} = \eta_k t_{\text{sens},k}^{\delta_k} + \zeta_k,
\end{equation}
where $\eta_k$ is the sensing efficiency coefficient, $\delta_k$ captures the nonlinearity of energy scaling, and $\zeta_k$ is a fixed energy offset such as booting cost\footnote{This parametric formulation allows for flexible modeling of heterogeneous sensing mechanisms. For example, we set $\delta_k = 1$, $\zeta_k = 0$ for cameras; $\delta_k = 2$, $\zeta_k = 0.5$ for radars; and $\delta_k = 1$, $\zeta_k = 0.1$ for lightweight environmental sensors.} \cite{bel2015energy,sharma2008optimal}. 
}

\subsubsection{Local Computation}
 While large datasets are essential for effective model training, transmitting substantial amounts of data from edge devices to the   server presents significant challenges due to constraints on bandwidth and energy resources. Furthermore, for certain applications, privacy concerns prohibit the outsourcing of data. As a result, distributed model training at the edge of the network has emerged as a promising solution. This method facilitates the cooperative training of DNNs across edge devices, eliminating the need to transfer local data to the  server. This has become increasingly famous with the integration of advanced mobile processors, equipped with graphics and neural processing units, into high-end devices.

In the Edge AI system,  each edge device performs local training using its own data for many iterations. Let $f_{\text{c},k}$ denote the computation capacity of edge device $k$, measured in CPU cycles per second. Denote by $t_{\text{comp}, k}$ the computation time of edge device $k$ at each iteration, which is given by
\begin{equation}
t_{\text{comp}, k}= \frac{C_kM_k}{f_{\text{c},k}} ,\label{a3}
\end{equation}
where  $C_k$ is the number of CPU cycles needed to process a sample. Moreover, the energy consumption for local computation  of device $k$ at each iteration is given by \cite{mao2016dynamic}

\begin{equation}
	E_{\text{comp}, k} = \kappa  C_kM_k f_{\text{c},k}^2,\label{a4}
\end{equation}
where $\kappa$ is the effective switched capacitance, which depends on the chip architecture.

\subsubsection{Wireless Communication}

{After the devices complete model training, their local models need to be transmitted to the server for aggregation. The server aggregates the local models from all devices to form a global model and broadcasts it to the devices. 
Let $t_{\text{comm},k}$ denote the time for edge device $k$ to transmit the model parameters at any iteration.. To capture the impact of channel variability (e.g., small-scale fading), we incorporate a deterministic channel quality factor $\theta_k > 0$ into the transmission model. Specifically, the achievable transmission rate for device $k$ at each iteration is given by
\begin{equation}
	r_{k} = b_k \log_2\left(1 + \frac{p_k \bar{g}_k \theta_k}{N_0 b_k} \right),
\end{equation}
where $b_k$ is the bandwidth allocated to device $k$, $p_k$ is the transmit power, $N_0$ is the noise power spectral density, and $\bar{g}_k$ denotes the average large-scale path gain. Moreover,  $\theta_k$ reflects device-specific small-scale fading or local propagation effects, and enables deterministic modeling of heterogeneous channel conditions.

The corresponding transmission time is given by
\begin{equation}
	t_{\text{comm},k} = \frac{d}{r_{k}},
\end{equation}
where $d$ is the dimension of the model parameters. The transmission energy  of each iteration is given by
\begin{equation}
	E_{\text{comm},k} = t_{\text{comm},k}  p_k.
\end{equation}}

\section{Problem Formulation}

\subsection{Convergence Analysis}

To derive an effective performance metric, we conduct an in-depth analysis of the considered Edge AI system. 
From a system-wide perspective, the learning performance should be closely associated with following two key aspects.
	\begin{itemize}
\item \textbf{The amount of data utilized for training}:   If an Edge AI system can access a greater amount of input data, the performance of the resulting model will  improve. Experimentally, the performance of models trained with different amounts of data can differ significantly. Specifically, \cite{krizhevsky2012imagenet}, \cite{sun2017revisiting} demonstrate that as the amount of data increases, deep learning models exhibit better generalization performance. They experiment by increasing the training data size for different models and observe significant improvements in accuracy.
Mathematically, the conclusion in \cite[Theorem 3.1]{friedlander2012hybrid} demonstrates that missing data increases the norm of the gradient update error during the learning process.

\item \textbf{The number of training rounds}: With a fixed amount of data, increasing the number of training rounds generally enhances the model’s performance. In the early stages of training, the performance gain is significant, as each additional round helps the model capture more complex patterns. However, as the number of training rounds continues to grow, the model gradually approaches a stable state, and the rate of performance improvement slows down. At this point, as analyzed in \cite{wang2020tackling} and \cite{reddi2020adaptive}, while additional rounds may still provide marginal gains, the benefit becomes minimal as the model has largely converged.  Therefore, it is crucial to balance the number of training rounds to maximize learning performance while avoiding unnecessary resource usage and mitigating model divergence.
	\end{itemize}

To quantitatively guide the system design and characterize the impact of data volume and training rounds, we conduct a convergence analysis of the Edge AI system. 
Our convergence analysis is based on the following assumptions:
\begin{assumption}\label{as1}
	\textbf{Lipschitz Continuity:} The function \( F(\cdot) \) is continuously differentiable, and its gradient \( \nabla F(\cdot) \) satisfies the Lipschitz continuity condition with a constant \( L > 0 \). Specifically, for any \( \xx, \yy \in \mathbb{R}^n \),
	\begin{align}
		\norm{\nabla F(\xx)-\nabla F(\yy)} \leq L \norm{\xx-\yy}.
	\end{align}
	This implies that the gradient does not change too rapidly, which guarantees  certain smoothness in  \( F(\cdot) \).
\end{assumption}

\begin{assumption}\label{as2}
	\textbf{Strong Convexity:} The function \( F(\cdot) \) is strongly convex with a parameter \( \mu > 0 \), meaning that for any \( \xx, \yy \in \mathbb{R}^n \),
	\begin{align}
		F(\xx) \geq F(\yy) + (\xx-\yy)^T \nabla F(\yy) + \frac{\mu}{2} \norm{\xx-\yy}^2.\label{12}
	\end{align}
	This condition ensures that the function exhibits a lower-bound curvature, leading to unique minimizers and a well-behaved optimization landscape.
\end{assumption}
\begin{assumption}\label{as3}
	\textbf{Bounded Gradient:} The gradient with respect to any training sample, $\nabla f(\cdot)$, is upper bounded by 
	\begin{align}
		\norm{f(\xx,{\xi}_m)}^2\leq \beta_1+\beta_2\norm{\nabla F(\xx)}^2,\forall m
	\end{align}
	for some constants $\beta_1 \geq 0$ and $\beta_2 \geq 0.$ This condition implies that the sample-wise gradient is bounded as a function of the true gradient of full sample. 
\end{assumption}

{Assumptions \ref{as1}-\ref{as3} are widely adopted in the literature on decentralized optimization and edge learning algorithms, see  \cite{koloskova2019decentralized,10818523,9451567,wang2018cooperative,10506083,zhai2025decentralized,zhai2024uav}.
Although these assumptions  may not strictly hold in practical deep learning models, especially for non-convex objectives such as those arising in convolutional neural networks (CNNs), they provide a widely accepted theoretical foundation for analyzing convergence behavior. 

Importantly, the goal of our convergence analysis is not to guarantee exact behavior under strict assumptions, but rather to derive a principled performance surrogate that guides resource allocation decisions. As shown in our experiments in Section~\ref{secV}, the derived convergence bound remains a reliable and consistent indicator of achievable learning performance, even when these assumptions are relaxed. Moreover, our optimization framework does not rely on explicit knowledge of constants such as $\mu$, $L$, or $\beta_2$, which further enhances its robustness in real-world deployment scenarios.}

\begin{theorem}\label{pro2}
	{Under Assumptions \ref{as1}-\ref{as3}, with  learning rate $\lambda=\frac{1}{L}\text{total}$, we have the following  ergodic convergence bound }
\begin{align}\label{421}
	&F({\xx_{I}}) - F(\xx^\ast) \leq \Psi(\mathcal{B})^I[F(\xx_0)-F(\xx)]\notag\\&+2\left[ \frac{M_{\textnormal{total}} - |\mathcal{B}|}{M_{\textnormal{total}}} \right]^2\frac{\beta_1}{L}\frac{1-\Psi(\mathcal{B})^I}{1-\Psi(\mathcal{B})}
\end{align}
	where \(\mathcal{B}= \{\mathcal{D}_k,\forall k \}\) is the set of overall collected data, and $\Psi(\mathcal{B})=\left[(1-\frac{\mu}{L})+4\left(\frac{M_{\textnormal{total}}-|\mathcal{B}|}{M_{\textnormal{total}}}\right)^2\beta_2\right]$.
\end{theorem}
\begin{proof}
	Please refer to Appendix A.
\end{proof}

{This theorem highlights that the convergence rate of Edge AI system  depends on the size of the collected dataset $|\mathcal{B}|$, and in turn, depends on  the sensing time $t_{\text{sens},k}$. A longer sensing time leads to more data, reducing $\left( \frac{M_{\textnormal{total}} - |\mathcal{B}|}{M_{\textnormal{total}}} \right)^2$ and subsequently tightening the bound. Not only does this accelerate  convergence but also improves generalization.}
With this convergence bound in \eqref{421}, we formulate problem to conduct system optimization in the next subsection.

\subsection{Problem Formulation}
According to the convergence analysis  in Theorem \ref{pro2}, we see that the  convergence bound on the right-hand side (RHS) of \eqref{421} converges  with the speed  $\Psi(\mathcal{B})^I$ when $\Psi(\mathcal{B})^I<1$. It can serve as a performance metric for optimizing the learning design. 
  Therefore, we jointly optimize each device's  the data collection time $\mathbf{t}_{\text{sens}}=[{t}_{\text{sens},1},{t}_{\text{sens},2}.\cdots,{t}_{\text{sens},K}]^\T$ and training rounds $I$ to minimize performance metric $\Psi(\mathcal{B})^I$. Taking the logarithm, we obtain
	\begin{subequations}\label{37}
	\begin{align}
		&\min_{\mathbf{t}_{\text{sens}},I}I\log\left(\alpha+\beta\left(\frac{M_{\text{total}}-|\mathcal{B}|}{M_{\text{total}}}\right)^2\right)\\
		&\text{s.t.}~\max_k\{t_{\text{sens},k}\}+I \max_k\{t_{\text{comp}, k}+t_{\text{comm}, k}\} \leq T_{\text{total}},\label{ttotal}\\
		&\quad~\,~~~\sum_{k=1}^{K}\left(E_{\text{sens}, k}+I	(E_{\text{comp}, k}+	E_{\text{comm}, k})\right)\leq E_{\text{total}},\label{etotal}
	\end{align}
\end{subequations}
{where $\alpha = 1 - \frac{\mu}{L}$ and $\beta = 4\beta_2$, which characterize the curvature of the loss function and the gradient variance due to data heterogeneity, respectively. A smaller $\alpha$ corresponds to stronger convexity (often associated with simpler models or better-conditioned data), while a larger $\beta$ arises in more heterogeneous or non-IID local datasets where local gradients deviate more from the global one.}

Here, constraint \eqref{ttotal} means that the total time spent on sensing, computation, and communication across all training rounds must not exceed the maximum time duration,  which  ensures that the Edge AI system operates within the allotted time; constraint \eqref{etotal} means that each device 
has a limited energy budget for all operations across training rounds, and this constraint ensures that the energy consumed by each device remains within its total available energy.

{Importantly, \eqref{37} inherently captures the trade-off between global model accuracy and device resource constraints: minimizing the convergence bound (i.e., improving accuracy) requires either more training rounds $I$ or a larger volume of collected data $|\mathcal{B}|$, but both directly increase the time and energy consumed by each device as constrained in~\eqref{ttotal} and~\eqref{etotal}. This balance ensures that the system can achieve optimal performance while respecting practical limitations on time and energy budgets at the device level.}

Problem  \eqref{37} seeks to balance the trade-offs among training time and energy consumption in Edge AI to optimize the learning performance. By minimizing the performance metric $\Psi(\mathcal{B})^T$, we can achieve efficient learning performance while adhering to the system's practical constraints.

Unfortunately, problem  \eqref{37} is a highly non-convex optimization problem, involving the joint optimization over multiple variables. Due to the complexity and interdependencies among the variables, finding the optimal solution is particularly challenging. In what follows, we approach the problem under two distinct scenarios: the case of homogeneous devices and the case of heterogeneous devices. By analyzing these two settings separately, we develop tailored algorithms and solution strategies that effectively exploit the specific characteristics of each scenario.


\section{System Optimization}
\subsection{Optimization in Homogeneous Scenario}
In this subsection, we consider the optimization problem under the scenario of homogeneous devices. In this case, each edge device shares identical specifications, including sensing frequency, sensing power, computational capability, switch capacitance, communication bandwidth, and power. This uniformity allows us to reduce the complexity of the original problem by aggregating variables and constraints,  leading to a more streamlined formulation. Specifically, the optimization reduces to finding a common solution that applies to all devices. In this case, problem \eqref{37} can be recast as
\begin{subequations}
	\begin{align}
		&\min_{{t}_{\text{sens}},I}I\log\left(\alpha+\beta\left(\frac{M_{\text{total}}-|\mathcal{B}|}{M_{\text{total}}}\right)^2\right)\\
		& ~~\text{s.t.}~~t_{\text{sens}}+I(t_{\text{comp}}+t_{\text{comm}}) \leq T_{\text{total}},\\
		&\quad~\,~~~K\left(E_{\text{sens}}+I(E_{\text{comp}}+	E_{\text{comm}})\right)\leq E_{\text{total}}.
	\end{align}
\end{subequations}

Here, the problem \eqref{37} reduces to optimizing only two variables $t_\text{sens}$ and $I$. By expressing the other parameters, including  $\mathcal{B}$, $t_{\text{comp}}$, $E_{\text{sens}}$ and $E_{\text{comp}}$, in  terms of 
$t_\text{sens}$, we have
\begin{subequations}\label{19}
	\begin{align}
		&\min_{{t}_{\text{sens}},I}I\log\left(\alpha+\beta\left(\frac{M_{\text{total}}-K{t}_{\text{sens}}f_\text{s}}{M_{\text{total}}}\right)^2\right)\\
		& ~~\text{s.t.}~~t_{\text{sens}}+I({{t}_\text{sens}}f_\text{s}C_k/{f_c}+t_{\text{comm}}) \leq T_{\text{total}},\label{30b}\\
		&\quad~\,~~~K\left(t_{\text{sens}}w+I( {t}_\text{sens}f_\text{s}\kappa  C_k  f_{\text{c}}^2+	E_{\text{comm}})\right)\leq E_{\text{total}}.\label{30c}
	\end{align}
\end{subequations}

By reorganizing constraints \eqref{30b} and \eqref{30c}, problem \eqref{19} can be rewritten as
\begin{subequations}\label{31}
	\begin{align}
		&\min_{{t}_{\text{sens}},I}I\log\left(\alpha+\beta\left(\frac{M_{\text{total}}-K{t}_{\text{sens}}f_\text{s}}{M_{\text{total}}}\right)^2\right)\label{31a}\\
		& ~~\text{s.t.}~~I\leq \frac{T_{\text{total}}-t_{\text{sens}}}{{{t}_\text{sens}}f_\text{s}C_k/{f_c}+t_{\text{comm}}},
		\label{310b}
		\\
		&\quad~\,~~~I\leq
		\frac{(E_{\text{total}}/K-t_{\text{sens}}w)}{{t}_\text{sens}f_\text{s}\kappa  C_k  f_{\text{c}}^2+	E_{\text{comm}}}.\label{310c}
	\end{align}
\end{subequations}
Constraints \eqref{310b}  and \eqref{310c} indicate that the number of training rounds $I$ needs to be smaller  than a certain value related to $t_{\text{sens}}$. Since  objective function \eqref{31a} is monotonically decreasing as $I$ increases, at least one of  constraints \eqref{310b} and \eqref{310c}  holds with equality at the optimal point. Therefore, the optimal solution of problem \eqref{31} can be found by solving the following  two problems. The first problem is given by
\begin{subequations}\label{311}
	\begin{align}
		&\min_{{t}_{\text{sens}}}\frac{T_{\text{total}}-t_{\text{sens}}}{{{t}_\text{sens}}f_\text{s}C_k/{f_c}+t_{\text{comm}}}\notag\\&\qquad\times\log\left(\alpha+\beta\left(\frac{M_{\text{total}}-K{t}_{\text{sens}}f_\text{s}}{M_{\text{total}}}\right)^2\right)\\
			& ~~\text{s.t.}~~{t}_{\text{sens}}\leq \frac{M_{\text{total}}}{Kf_\text{s}}.
	\end{align}
\end{subequations}

The second problem is given by
\begin{subequations}\label{333}
	\begin{align}
		&\min_{{t}_{\text{sens}}}	\frac{(E_{\text{total}/K}-t_{\text{sens}}w)}{{t}_\text{sens}f_\text{s}\kappa  C_k  f_{\text{c}}^2+	E_{\text{comm}}}\notag\\&\qquad\times\log\left(\alpha+\beta\left(\frac{M_{\text{total}}-K{t}_{\text{sens}}f_\text{s}}{M_{\text{total}}}\right)^2\right)\\
		& ~~\text{s.t.}~~{t}_{\text{sens}}\leq \frac{M_{\text{total}}}{Kf_\text{s}}.
	\end{align}
\end{subequations}

Since \eqref{311} and \eqref{333} are non-convex with a bounded optimization variable, we can solve these two problems using  one-dimensional search. After solving the problems, the solution that minimizes the objective function between problems  \eqref{311} and \eqref{333} is selected as the optimal solution.

The complexity of this algorithm is primarily determined by the  one-dimensional search method, which is $\mathcal{O}(M)$, where $M$ represents  the number of search variables.

\subsection{Optimization in Heterogeneous Scenario}

In real-world edge AI systems, devices typically possess varying sensing, computation, and communication capabilities. This heterogeneity can arise from differences in hardware specifications, energy constraints, or network conditions. In this section, we extend our analysis to a more generalized scenario where  devices are heterogeneous. Define $\mathbf{t}_{\text{sens}}=[{t}_{\text{sens},1},\cdots,{t}_{\text{sens},K}]^\T$ and $\ff_\text{s}=[f_{s,1},\cdots,f_{s,K}]$, then our problem becomes
\begin{subequations}\label{34}
	\begin{align}
		&\min_{\mathbf{t}_{\text{sens}},I}I\log\left(\alpha+\beta\left(\frac{M_{\text{total}}-\mathbf{t}_{\text{sens}}^\T\ff_{\text{s}}}{M_{\text{total}}}\right)^2\right)\label{34a}\\
			& \text{s.t.}~\max_k\{{t_{\text{sens},k}}\}+I\max_k\{{{t}_{\text{sens},k}}f_{\text{s},k}C_k/{f_{\text{c},k}}+t_{\text{comm},k}\}\notag\\&\qquad\qquad\leq T_{\text{total}},\label{34b}\\
	&\quad~\,~\sum_{k=1}^{K}\left(t_{\text{sens},k}w_k+I( {t}_{\text{sens},k}f_{\text{s},k}\kappa  C_k  f_{\text{c},k}^2+	E_{\text{comm},k})\right)\notag\\&\qquad\qquad\leq E_{\text{total}}.\label{34c}
	\end{align}
\end{subequations}

This problem is  highly non-convex, with both the objective function \eqref{34a} and  constraints \eqref{34b} and \eqref{34c} exhibiting non-convexity. To address this challenge, we adopt an alternating optimization approach \cite{bezdek2002some}, where the number of iterations $I$ and the sensing time $\mathbf{t}_{\text{sens}}$ are decoupled. By doing so, we decompose the original problem into two alternately solved  subproblems. This strategy allows us to iteratively optimize one variable while holding the other fixed, progressively refining the solution at each step. The resulting subproblems are presented as follows.
\subsubsection{Optimizing $I$ with fixed $\mathbf{t}_{\text{sens}}$}

 Given fixed  $\mathbf{t}_{\text{sens}}$,  problem \eqref{34} reduces to focusing only on the variable $I$. The simplified problem can be expressed as 
\begin{subequations}\label{35}
	\begin{align}
		&\min_{I}I\log\left(\alpha+\beta\left(\frac{M_{\text{total}}-|\mathbf{t}_{\text{sens}}^\T\ff_{\text{s}}|}{M_{\text{total}}}\right)^2\right)\label{3434a}\\
		& \text{s.t.}~~\max_k\{{t_{\text{sens},k}}\}+I\max_k\{{{t}_{\text{sens},k}}f_{\text{s},k}C_k/{f_{\text{c},k}}+t_{\text{comm},k}\}\notag\\&\qquad\qquad\leq T_{\text{total}},\label{334b}\\
		&\quad~\,~~\sum_{k=1}^{K}\left(t_{\text{sens},k}w_k+I( {t}_{\text{sens},k}f_{\text{s},k}\kappa  C_k  f_{\text{c},k}^2+	E_{\text{comm},k})\right)\notag\\&\qquad\qquad\leq E_{\text{total}}.\label{343c}
	\end{align}
\end{subequations}

Since $\log\left(\alpha+\beta\left(\frac{M_{\text{total}}-|\mathbf{t}_{\text{sens}}^\T\ff_{\text{s}}|}{M_{\text{total}}}\right)^2\right)<0$, this subproblem can be further simplified to maximize $I$. The optimization task now revolves around finding the biggest number of rounds $I$ that satisfies the system constraints. This leads to the following transformation of  \eqref{35}
\begin{subequations}\label{36}
	\begin{align}
		&\max\quad  I\label{34a34a}\\
		& \text{s.t.}~~\max_k\{{t_{\text{sens},k}}\}+I\max_k\{{{t}_{\text{sens},k}}f_{\text{s},k}C_k/{f_{\text{c},k}}+t_{\text{comm},k}\}\notag\\&\qquad\qquad\leq T_{\text{total}},\label{3a34b}\\
		&\quad~\,~~~\sum_{k=1}^{K}\left(t_{\text{sens},k}w_k+I( {t}_{\text{sens},k}f_{\text{s},k}\kappa  C_k  f_{\text{c},k}^2+	E_{\text{comm},k})\right)\notag\\&\qquad\qquad\leq E_{\text{total}}.\label{343ac}
	\end{align}
\end{subequations}

Problem \eqref{36} is  convex, as all the constraints involving 
$I$ are linear. Therefore, solving this is straightforward. We only need to check the constraints \eqref{3a34b} and \eqref{343ac} to find the maximum possible value of $I$.
\subsubsection{Optimizing  $\mathbf{t}_{\textnormal{sens}}$ with fixed $I$}
Given fixed $I$, problem \eqref{34} simplifies, as we now only need to focus on the remaining optimization variables related to the sensing time $\mathbf{t}_{\text{sens}}$. In this case, problem \eqref{34} reduces to
\begin{subequations}\label{374}
	\begin{align}
		&\min_{\mathbf{t}_{\text{sens}}}\log\left(\alpha+\beta\left(\frac{M_{\text{total}}-\mathbf{t}_{\text{sens}}^\T\ff_{\text{s}}}{M_{\text{total}}}\right)^2\right)\label{347a}\\
		& \text{s.t.}~~\max_k\{{t_{\text{sens},k}}\}+I\max_k\{{{t}_{\text{sens},k}}f_{\text{s},k}C_k/{f_{\text{c},k}}+t_{\text{comm},k}\}\notag\\&\qquad\qquad\leq T_{\text{total}},\label{374b}\\
		&\quad~\,~~~\sum_{k=1}^{K}\left(t_{\text{sens},k}w_k+I( {t}_{\text{sens},k}f_{\text{s},k}\kappa  C_k  f_{\text{c},k}^2+	E_{\text{comm},k})\right)\notag\\&\qquad\qquad\leq E_{\text{total}}.\label{374c}
	\end{align}
\end{subequations}

The objective function \eqref{347a} is non-convex. However, considering that the logarithmic term $\log\left(\alpha+\beta\left(\frac{M_{\text{total}}-\mathbf{t}_{\text{sens}}^\T\ff_{\text{s}}}{M_{\text{total}}}\right)^2\right)$ is decreases monotonically  with respect to $\mathbf{t}_{\text{sens}}^\T\ff_{\text{s}}$.
Thus, we transform the original non-convex problem  \eqref{374} into a simpler optimization problem that maximizes $\mathbf{t}_{\text{sens}}^\T\ff_{\text{s}}$
under the given constraints \eqref{374b} and \eqref{374c}.  Problem \eqref{374} can now be rewritten as
\begin{subequations}\label{374s}
	\begin{align}
		&\max_{\mathbf{t}_{\text{sens}}}\quad \mathbf{t}_{\text{sens}}^\T\ff_{\text{s}}\label{347sa}\\
		& \text{s.t.}~~\max_k\{{t_{\text{sens},k}}\}+I\max_k\{{{t}_{\text{sens},k}}f_{\text{s},k}C_k/{f_{\text{c},k}}+t_{\text{comm},k}\}\notag\\&\qquad\qquad\leq T_{\text{total}},\label{374sb}\\
		&\quad~\,~\sum_{k=1}^{K}\left(t_{\text{sens},k}w_k+I( {t}_{\text{sens},k}f_{\text{s},k}\kappa  C_k  f_{\text{c},k}^2+	E_{\text{comm},k})\right)\notag\\&\qquad\qquad\leq E_{\text{total}}.\label{374sc}
	\end{align}
\end{subequations}

Note that the  objective function \eqref{347sa} and constraint \eqref{374sc} are linear with respect to $\mathbf{t}_{\text{sens}}$, and thus they are convex functions.  Therefore, problem \eqref{374s}  is a standard convex optimization problem and can be efficiently solved using existing convex optimization tools, such as CVX \cite{grant2014cvx}. The convexity ensures that these tools can find the global optimum for \eqref{374s}, providing a straightforward way to solve the optimization formulation.
\subsubsection{Alternating Optimization Iteration}
To solve the non-convex optimization problem \eqref{34}, we utilize an alternating optimization method, which iteratively optimizes over the variables $I$ and $\mathbf{t}_{\text{sens}}$.  We summarize the proposed alternating process in Algorithm \ref{alg}.
	\begin{algorithm}[h]
	\caption{Alternating Algorithm for Solving \eqref{34}}
	\label{alg} 
	\begin{algorithmic}[1] 
		\STATE {\textbf{Initialize:}} The initial values $\mathbf{t}_{\text{sens}}^{(0)}$ and $I^{(0)}$,  iteration index $t=1$,  predefined  precision $\epsilon$.
		\REPEAT 
		\STATE With given $\mathbf{t}_{\text{sens}}^{(t)}$, compute $I^{(t+1)}$ based on \eqref{36}.
		\STATE With given $I^{(t+1)}$, compute $\mathbf{t}_{\text{sens}}^{(t+1)}$ by solving problem \eqref{374s}.
		\STATE Calculate  the value $L^{(t+1)}$ of objective function in \eqref{34}.
		\STATE $t = t + 1$.
		\UNTIL $\left|(L^{(t)} - L^{(t-1)})/L^{(t-1)}\right| \leq \epsilon$
		\ENSURE $\{I,\mathbf{t}_{\text{sens}}\}$.
	\end{algorithmic}
\end{algorithm}

We now analyze the computational complexity of {Algorithm \ref{alg}}. The primary contributor to the complexity is Step 4, where we use the convex optimization tool CVX \cite{grant2014cvx}. The standard CVX tool implements an interior-point method, which incurs the complexity of $\mathcal{O}(n^{3.5}\log(1/\varepsilon))$, 
where $n$ denotes the number of optimization variables and $\varepsilon$ is the target accuracy of the solution. 
As a result, the complexity of Step~4 is $\mathcal{O}(K^{3.5})$, where $K$ is the number of edge devices.
Therefore, the total complexity of the algorithm is $\mathcal{O}(I_{\text{loop}}K^{3.5})$, where $I_{\text{loop}}$ represents the number of iterations required to meet the precision criteria. Given the relatively low complexity, the proposed algorithm is computationally efficient and well-suited for fast system deployment and optimization in real-world edge AI scenarios.

The proposed alternating optimization follows the block coordinate descent (BCD) framework for solving problem~\eqref{34}. 
Specifically, let the objective function of~\eqref{34} be denoted as 
\[
L(\mathbf{t}_{\text{sens}}, I) = I \log\!\left(\alpha + \beta\!\left(\frac{M_{\text{total}} - \mathbf{t}_{\text{sens}}^{T}\mathbf{f}_{s}}{M_{\text{total}}}\right)^{\!2}\right).
\]
At iteration $t$, given the current variables $\{\mathbf{t}_{\text{sens}}^{(t)}, I^{(t)}\}$, the algorithm performs the following two updates:
\begin{align}
	I^{(t+1)} &= \arg\min_{I}\; L(\mathbf{t}_{\text{sens}}^{(t)}, I), \label{eq:updateI}\\
	\mathbf{t}_{\text{sens}}^{(t+1)} &= \arg\min_{\mathbf{t}_{\text{sens}}}\; L(\mathbf{t}_{\text{sens}}, I^{(t+1)}), \label{eq:updatetsens}
\end{align}
subject to the corresponding time and energy constraints in~\eqref{34b}–\eqref{34c}. 
Each subproblem~\eqref{eq:updateI} and~\eqref{eq:updatetsens} is convex and can be optimally solved. 
Therefore, following the standard BCD convergence theory~\cite{razaviyayn2013unified}, the objective value satisfies
\[
L(\mathbf{t}_{\text{sens}}^{(t+1)}, I^{(t+1)}) 
\le L(\mathbf{t}_{\text{sens}}^{(t)}, I^{(t)}),
\]
which indicates that the sequence $\{L^{(t)}\}$ is monotonically non-increasing. 
Consequently, the alternating updates converge to a stationary point of~\eqref{34}.

\section{Numerical Results}\label{secV}

\subsection{Parameter Settings}
{We perform an image classification task using the CIFAR-10 dataset  \cite{krizhevsky2009learning}. From the original dataset, we set that there are totally 20,000 samples in the environment for model training (i.e., $M_\text{total}=20000$) and use 10,000 samples for accuracy validation. We divide the $K=20$ edge devices into 10 equally sized groups, and each group is assigned data samples exclusively from one  class. For the neural network configuration, we utilize a convolutional neural network (CNN)\footnote{Our system model and optimization framework are also applicable to large-scale models such as LLMs \cite{naveed2023comprehensive}, although additional system-level considerations (e.g., memory, latency, and model compression) may be needed to enable practical deployment.
	} comprising two convolutional layers (5 × 5 kernels) with 10 and 20 filters, respectively, each followed by 2 × 2 max-pooling layers. After these layers, a batch normalization layer is added, followed by a fully connected layer with 50 units, using ReLU as the activation function, and a softmax output layer. In total, the model contains 21,880 parameters, and the cross-entropy loss function is employed for training. The model is trained using an NVIDIA RTX 3060Ti GPU.

In our Edge AI system setup, we consider a configuration consisting of $K=20$ edge devices and a single edge server.  The path loss model is  $40 + 30\log_{10}(\gamma)$, where $\gamma$ is the distance in meters, the effective  switched capacitance is set to $\kappa=10^{-11}$, and the noise power spectral density is set to $N_0=4 \times 10^{-11}$ W/Hz. The optimization parameters are chosen as $\alpha = 0.5$ and $\beta = 0.5$ for performance evaluation. For the homogeneous device scenario, each device operates with a sampling rate of $f_{\text{s},k}=10$ Hz. The energy consumption coefficient $\eta_k=0.5,~{\delta_k}=1,~  \zeta_k=0.1$. The average large-scale path gain $\bar{g}_k$ is set to  $1$ for simplicity. The processing of each sample requires $C_k=50$ CPU cycles and each device is equipped with a computational capability of $f_{\text{c},k}=200$ Hz (CPU cycles per second). Additionally, the transmission power of each device is $p_k=0.5$ W,  the channel bandwidth is $b_k=2000$ Hz, the training time constraint is $T_{\text{total}}=5000$ s, and the overall energy constraint is $E_{\text{total}}=12000$ J.  Devices are located at a  distance of $d_k=50$ m, $\forall k$ from the server. For the heterogeneous device scenario, the parameter values of different device are generated according to a Gaussian distribution, where the mean is set to the corresponding parameter value used in the homogeneous device scenario, and the standard deviation is set to $0.5$ times the respective mean.}

\subsection{Verification of Objective Function}
In this paper, our objective function is derived based on the convergence analysis. To validate the effectiveness of the proposed objective function  (and the accuracy of the derived convergence analysis), we modify the configuration of the Edge AI system to study the relationship between system performance and the objective function given in \eqref{37}. 

We first evaluate the Edge AI system by running it under different numbers of training rounds. Each edge device is configured to collect $750$ samples for local training. In Table \ref{tab1}, we present the corresponding values of the objective function, accuracy, and loss function across different training rounds. As observed, there is a  monotonic relationship between the objective function and learning performance. Specifically, as the objective function value decreases, the learning accuracy progressively improves, while the loss function steadily decreases. 
\begin{table}[h]
	
	\centering
		\caption{Training Results with Different Rounds}
		\label{tab1}
	\begin{tabular}{cccc}
		\toprule
		\textbf{Rounds} & \textbf{Objective} & \textbf{Accuracy} & \textbf{Loss}  \\
		\midrule
20 & -12.6505 &  0.3622 & 0.5822 \\
30 &  -18.9757 & 0.5132 & 0.5216 \\
50 &  -31.6261 & 0.6215 & 0.4125 \\
100 &  -63.2523 & 0.7300 & 0.3114\\
150 & -94.8784 &  0.8050 & 0.2522 \\
		\bottomrule
	\end{tabular}
\end{table}

\begin{figure}[h]
	
	\centering            
	\includegraphics[width=0.9\linewidth]{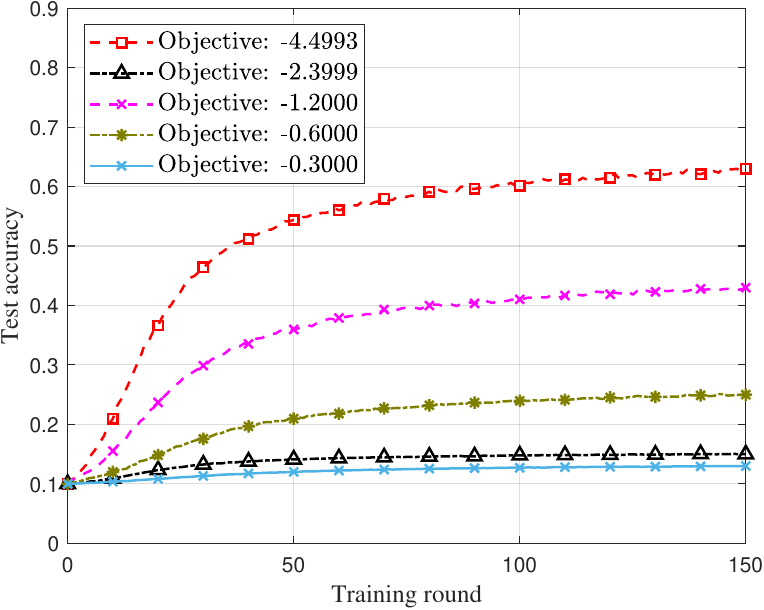}
	\caption{Test accuracy versus training round.}
	\label{diff_r_aver}
\end{figure}
\begin{figure}[h]
		
	\centering                    
	\includegraphics[width=0.9\linewidth]{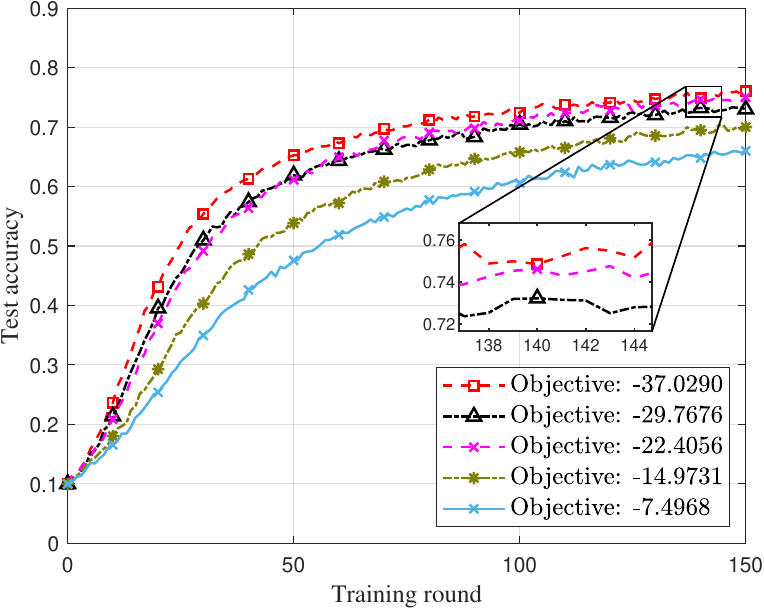}   
	\caption{Test accuracy versus training round.}
	\label{diff_r_min}
\end{figure}
{To further investigate how data quantity affects model performance, we conduct experiments under varying total sample sizes. In particular, we divide the experiments into two groups for clarity.
	
	Fig.~\ref{diff_r_aver} illustrates the test accuracy curves when the total number of training samples is set to 40, 80, 160, 320, and 600. We observe that as the sample size increases, the objective value improves significantly, leading to faster convergence and higher final accuracy. This figure highlights the sensitivity of learning performance in data-scarce regimes, where a modest increase in data volume can yield noticeable improvements.
	
	Fig.~\ref{diff_r_min} shows the results for larger sample sizes: 1000, 2000, 3000, 4000, and 5000. The performance consistently improves as more data becomes available, confirming the monotonic relationship between the optimized objective and the model accuracy. Even in the high-data regime, models benefit from increased data volume, albeit with diminishing returns in later stages.
	
	Overall, these two figures jointly validate our formulation: a smaller (better) objective function value corresponds to more efficient data acquisition and improved model accuracy, thereby demonstrating the strong correlation between the optimization objective and final learning performance.}

\subsection{Simulation Under Different Scenarios}

In this section, we explore the impact of varying key system parameters on the learning performance of the Edge AI system. First, we investigate how changes in the parameters $\alpha$ and $\beta$, which are introduced in the convergence analysis, affect the system's performance. These parameters represent assumptions made during the convergence analysis, and in different deployment scenarios, the values of $\alpha$ and $\beta$ that satisfy the convergence requirements may vary.

To assess their impact, we conduct a series of experiments in homogeneous devices case where we modify the values of $\alpha$ and $\beta$, and subsequently perform Edge AI optimization and training. This analysis helps to better understand the sensitivity of the system's learning capabilities to changes in these key assumptions.

\begin{table}[h]
	
	\centering
	\caption{Training Results for Different Alpha and Beta Values}
	\label{tabss2}
	\begin{tabular}{ccccccc}
		\toprule
		$\alpha$ & $\beta$ & \textbf{Obj.} & $t_{\text{sens}}$ & \textbf{Rounds} & \textbf{Acc.} & \textbf{Loss} \\
		\midrule
		0.50 & 0.50 & -19.0408 & 25 & 77  & 0.7485 & 0.2650 \\
		0.40 & 0.60 & -23.6003 & 34 & 57  & 0.7080 & 0.3120 \\
		0.60 & 0.40 & -14.8917 & 19 & 100 & 0.6753 & 0.3850 \\
		0.80 & 0.20 & -7.2075  & 14 & 134 & 0.6384 & 0.4560 \\
		0.20 & 0.80 & -36.3213 & 66 & 29  & 0.6031 & 0.5175 \\
		\bottomrule
	\end{tabular}
\end{table}

{
The results in Table~\ref{tabss2} illustrate the impact of varying $\alpha$ and $\beta$ on the Edge AI system’s performance. As $\alpha$ increases (e.g., from 0.20 to 0.80), the system places more emphasis on training rounds, as reflected by the reduced data collection time $t_{\text{sens}}$ and the increased number of training rounds. For example, when $\alpha = 0.80$ and $\beta = 0.20$, $t_{\text{sens}}$ drops to 14, while the training rounds increase to 134.

Conversely, as $\beta$ increases, the system prioritizes data collection, resulting in a longer sensing time and fewer training rounds due to energy constraints. For instance, with $\beta = 0.80$ and $\alpha = 0.20$, $t_{\text{sens}}$ rises to 66, while the number of training rounds decreases to 29.

Notably, the configuration $\alpha = 0.50$ and $\beta = 0.50$ achieves the best trade-off between accuracy and efficiency, attaining the highest test accuracy of 0.7485 and the lowest training loss of 0.2650. This suggests that a balanced weighting between data collection and training leads to the most effective system performance.

We next investigate the impact of device heterogeneity, specifically the degree of heterogeneity (measured by the standard deviation), on the performance of the Edge AI system. To evaluate this, we conduct system training in heterogeneous devices case by varying the standard deviation. The results of the  $20$ devices are presented.

Fig.~\ref{sensing_time}  compares the sensing time for different devices under three levels of standard deviation. As the standard deviation increases, the variance in sensing time among devices becomes more pronounced, indicating that devices with higher heterogeneity exhibit a wider range of sensing performance.
Fig.~\ref{time_cons} highlights the total time consumption for each device under different standard deviations. It shows that when device heterogeneity is low (i.e., with a lower standard deviation), the time consumption is more consistent across devices. As the standard deviation increases, the disparity in time consumption becomes more significant.
In Fig.~\ref{energy_consumption} the energy consumption plot demonstrates that devices with lower heterogeneity consume energy more consistently, while higher heterogeneity leads to a broader range of energy consumption across devices. Moreover, Table \ref{tabsscc2} shows that when the standard deviation is low, the system achieves higher test accuracy  and lower training loss . In contrast, higher heterogeneity (standard deviations of 0.3 and 0.5) results in significantly lower accuracy and higher loss, indicating degraded learning performance.

The overall learning performance is better when the level of device heterogeneity is lower. This improvement can be attributed to the more balanced specifications of the devices, where the capabilities of each device are more aligned. In such cases, the system is not constrained by the limitations of less capable devices, leading to more efficient learning. Conversely, as heterogeneity increases, the performance of the system is dragged down by the weaker devices, resulting in poorer learning outcomes. This reinforces the importance of maintaining uniformity in device specifications for optimal Edge AI system performance.}
\begin{figure}[h]
	\centering            
	\includegraphics[width=0.8\linewidth]{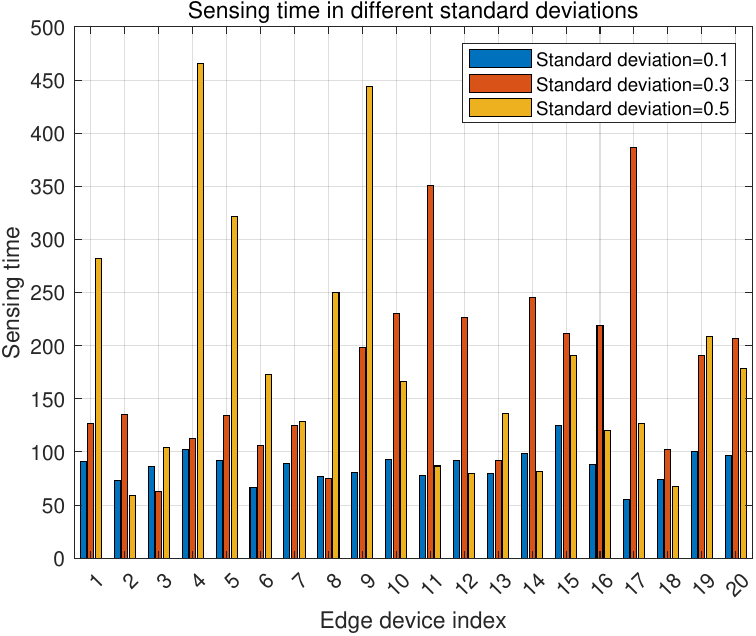}
	\caption{Sensing time in different standard deviations.}
	\label{sensing_time}
\end{figure}
\begin{figure}[h]
	\centering                    
	\includegraphics[width=0.8\linewidth]{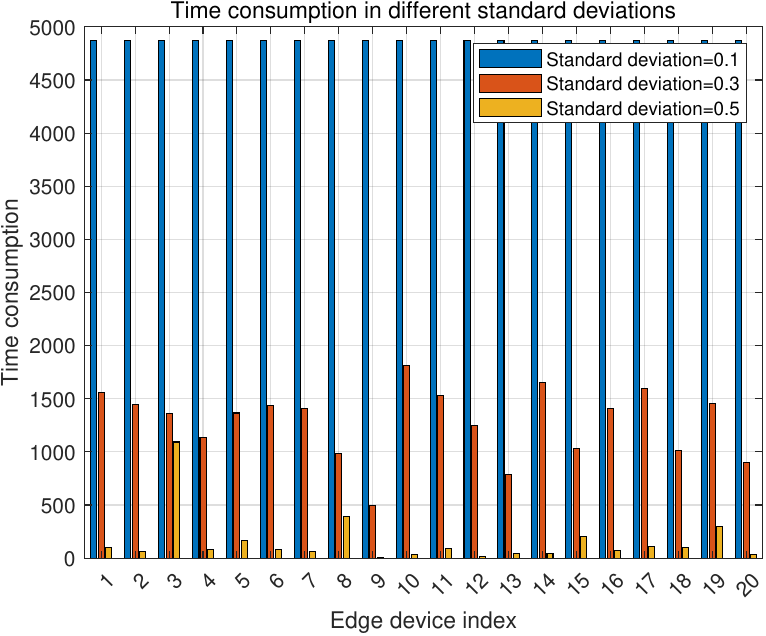}   
	\caption{Time consumption in different standard deviations.}
	\label{time_cons}
\end{figure}\begin{figure}[h]
\centering            
\includegraphics[width=0.8\linewidth]{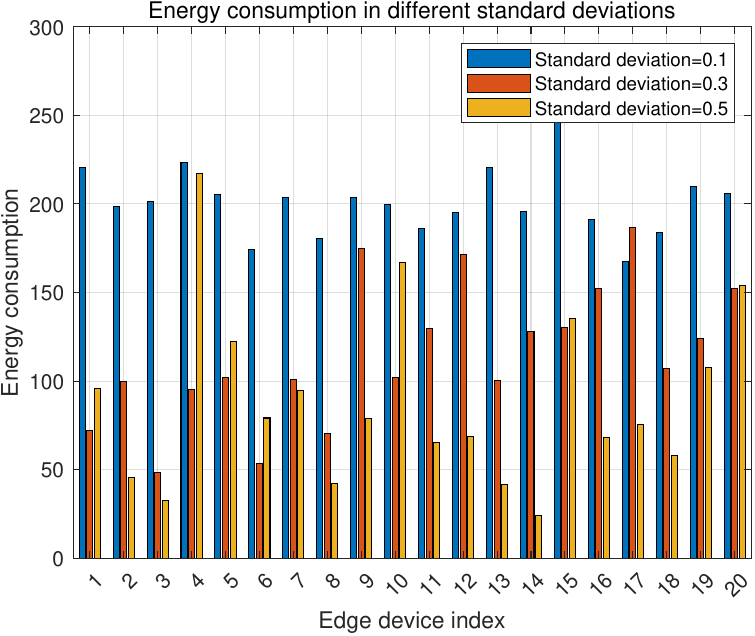}
\caption{Energy consumption in different standard deviations.}
\label{energy_consumption}
\end{figure}

\begin{table}[h]
	
	\centering
	\caption{Training Results for Different Standard Deviations}
	\label{tabsscc2}
	\begin{tabular}{ccccc}
		\toprule
		\textbf{Std. Dev.} & \textbf{Obj.} & \textbf{Rounds} & \textbf{Acc.} & \textbf{Loss} \\
		\midrule
		0.1 & -16.4252 & 74  & 0.7288 & 0.2785 \\
		0.3 & -9.7236  & 27  & 0.5102 & 0.6057 \\
		0.5 & -12.5621 & 4   & 0.1123 & 1.8658 \\
		\bottomrule
	\end{tabular}
\end{table}

\subsection{Performance Comparison with Benchmarks}

{
	In this subsection, we evaluate the proposed scheme against several state-of-the-art baseline strategies under heterogeneous device setting. These methods reflect diverse approaches to energy and time efficient data acquisition in Edge AI systems.

\begin{itemize}
	
	\item \textbf{Efficiency-Aware Sampling Strategy (EAS) \cite{zeng2021energy,chen2020joint}:}
	
	We design an efficiency-aware sampling strategy  that assigns sensing durations to devices based on a newly defined {composite energy efficiency factor} $\Phi_k$. This factor captures the device's efficiency across sensing, computation, and communication, and is computed as
	\[
	\Phi_k = \frac{f_{\text{s},k}  f_{\text{c},k} r_k}{c_k  (\eta_k + \zeta_k + \kappa_k + p_k)},
	\]
	where $f_{s,k}$ denotes the sampling rate of device $k$, $f_{\text{s},k}$ is the base computation throughput per sample of device $k$, and $r_k$ is the uplink transmission rate of device $k$. The denominator quantifies energy cost: $c_k$ is the per-sample computation complexity, $\eta_k$ is the sensing efficiency, $\zeta_k$ is the sensing booting overhead, $\kappa_k$ is the computation energy coefficient, and $p_k$ is the transmission power.
	
	Each device $k$ is then assigned a sensing time proportional to its relative efficiency:
	\[
	t_{\text{sens},k} = T_{\text{sens}}^{\max} \cdot \frac{\Phi_k}{\max_j \Phi_j},
	\]
	where $T_{\text{sens}}^{\max}$ denotes the total available sensing time. The value of $T_{\text{sens}}^{\max}$ is selected such that the entire environment dataset of size $M_{\text{total}}$ can be jointly collected in parallel by all devices. 
	 This method ensures efficient use of sensing budgets and reflects realistic energy-aware federated learning frameworks \cite{zeng2021energy,chen2020joint}.
	
	\item \textbf{Greedy Sampling Strategy (GSS) \cite{singhal2024greedy,jeong2018communication}:}
	In this approach, all devices simultaneously acquire data until the entire environmental dataset of size $M_{\text{total}}$ is collected. No prioritization or coordination is applied across devices. This strategy mirrors practical deployment scenarios with non-optimized data acquisition \cite{singhal2024greedy,jeong2018communication}.
	
	\item \textbf{Partial Sampling Strategy (PSS) \cite{tran2019federated,yang2020energy}:}
	To model scenarios with strict sensing budgets, this strategy limits the collected data to $15\%$ of the total sample volume. Based on this, we calculate the corresponding sensing time \(\mathbf{t}_{\text{sens}}\), and then substitute \(\mathbf{t}_{\text{sens}}\) into \eqref{37} to solve for the optimal number of training rounds \(I\). This baseline is widely adopted in resource-constrained federated learning literature \cite{tran2019federated,yang2020energy}.
	
	\item \textbf{Unconstrained Resource Strategy (URS):}
	This idealized baseline assumes unlimited energy and time resources. All edge devices collaboratively collect the full environment dataset, and training proceeds with sufficient iterations to achieve the best possible model performance. The results of this scheme serve as a performance upper bound for comparison purposes.
	
\end{itemize}}

\begin{figure}[h]
	
	\centering            
	\includegraphics[width=0.9\linewidth]{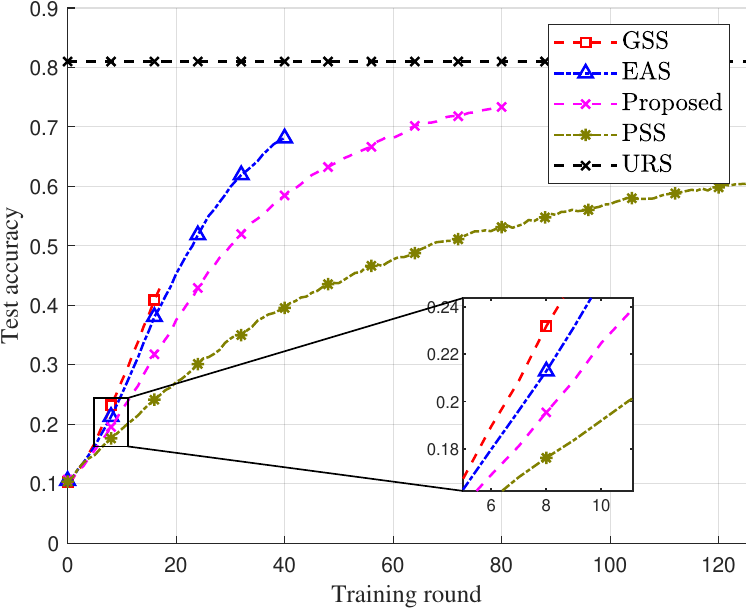}
	\caption{Test accuracy versus training round on CIFAR-10 dataset.}
	\label{homo_diff_losspdf}
\end{figure}
\begin{figure}[h]
		
	\centering                    
	\includegraphics[width=0.9\linewidth]{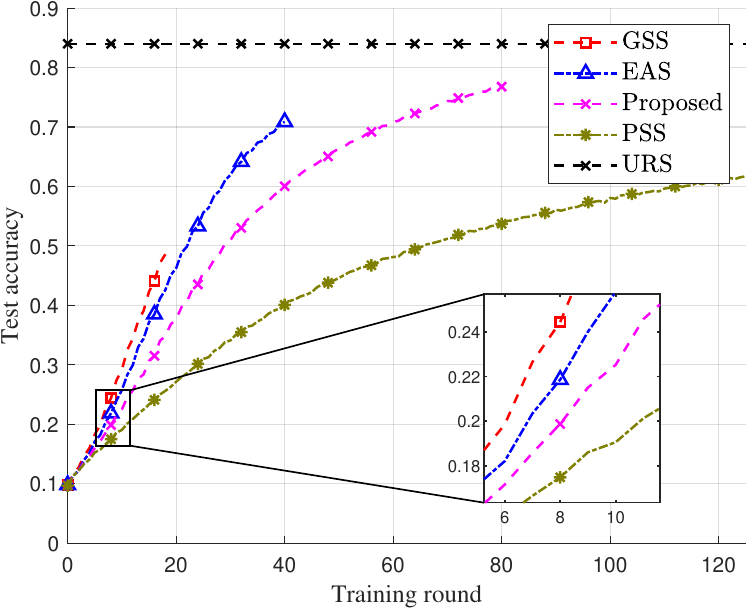}   
	\caption{Test accuracy versus training round on EuroSAT dataset.}
	\label{homo_diff_loss}
\end{figure}
{
To validate the effectiveness of the proposed scheme, we evaluate its performance on both the CIFAR-10 and EuroSAT datasets under heterogeneous device settings, comparing it with four representative baseline strategies: GSS, EAS, PSS, and URS. The test accuracy curves versus training rounds are illustrated in Fig.~\ref{homo_diff_losspdf} and Fig.~\ref{homo_diff_loss}.

On the CIFAR-10 dataset (Fig.~\ref{homo_diff_losspdf}), the proposed scheme  outperforms the baseline methods in the entire training process. While GSS and EAS demonstrate relatively rapid accuracy gains in the initial rounds, their performance eventually saturates due to inefficient utilization of training resources or lack of joint optimization. In contrast, the proposed method achieves a more balanced trade-off between data acquisition and training iterations, leading to sustained improvement and higher final accuracy. Specifically, although EAS leverages energy-efficient sensing on favorable devices, it does not adaptively allocate training time, resulting in suboptimal final performance. PSS performs the worst among all schemes due to its severely limited data usage. GSS starts strong but fails to maintain long-term gains due to insufficient training rounds. As expected, the URS baseline achieves the highest performance, serving as a theoretical upper bound assuming unlimited resources. Nevertheless, our proposed scheme approaches this ideal performance under realistic constraints.

A similar performance pattern is observed on the EuroSAT dataset (Fig.~\ref{homo_diff_loss}). The proposed method outperforms all baseline strategies in terms of final test accuracy. Compared to EAS, the proposed scheme exhibits better long-term accuracy, especially as the training progresses. This observation suggests that the advantage of balanced resource allocation becomes more pronounced when dealing with spatially redundant remote sensing data. The PSS and GSS baselines again suffer from either underutilization of training time or excessive data acquisition, leading to inferior performance. Overall, the proposed method demonstrates strong generalization ability across different datasets.

We observe that methods with longer sensing durations, such as GSS and EAS, achieve faster initial accuracy improvements due to access to a larger volume of training data. Specifically, GSS exhibits the steepest accuracy growth since it collects the entire dataset, followed by EAS, which allocates more sensing time to energy-efficient devices and thus accumulates a substantial amount of data. In contrast, the proposed scheme achieves a balanced trade-off between sensing and training: although it collects fewer data samples than GSS and EAS, it reserves sufficient time for iterative training, resulting in more stable and consistent convergence. On the other hand, PSS suffers from extremely limited data acquisition, which hinders its ability to support effective training and leads to slower accuracy improvement throughout.
}

	\section{Conclusions}{
In this paper, we studied the learning performance optimization for Edge AI systems under time and energy constraints. We proposed a joint optimization framework that allocates sensing time and training rounds across devices to minimize a convergence-based surrogate objective. Our formulation captures key resource trade-offs in sensing, computation, and communication, and accommodates both homogeneous and heterogeneous device settings. Theoretical convergence analysis was provided, and efficient algorithms were developed to solve the resulting non-convex problem. Extensive experiments on CIFAR-10 and EuroSAT datasets validate the effectiveness of the proposed approach in improving test accuracy under practical constraints.
Future extensions include incorporating data quality, multimodal sensing, or active perception into the optimization framework, which are important in more complex real-world Edge AI systems beyond the scope of our current model.
}

	\appendices

\section{Proof of Theorem~\ref{pro2}}\label{app_a}
In edge AI system, the aggregated gradient in \eqref{eq04} is an approximation of gradient $\nabla F({\xx_i})$, i.e.,
\begin{align}
	{\gg_i}=\nabla F({\xx_i})+{\ee_i},\label{11b}
\end{align}
where ${\ee_i}$ is the residual error in gradient evaluation.
From Assumption \ref{as1} and \ref{as2}, the following inequalities can be established:
\begin{align}
	F(\yy ) &\leq F(\xx) + (\yy - \xx)^\top \nabla F(\xx) + \frac{L}{2}\|\yy - \xx\|^2, \label{10}\\
	F(\yy ) &\geq F(\xx) + (\yy  - \xx)^\top \nabla F(\xx) + \frac{\mu}{2}\|\yy  - \xx\|^2. \label{13}
\end{align}

Let \( \xx = {\xx_i} \) and \( \yy = {\xx_i} - (1/L){\gg_i} \) in inequality \eqref{10}, and simplify to obtain:
\begin{align*}
	F({\xx_i} - (1/L)\gg_i) &\leq F({\xx_i}) - \frac{1}{L} \gg_i^\top \nabla F({\xx_i}) \\
	&\quad + \frac{1}{2L}\|\gg_i\|^2.
\end{align*}

Next, substituting the definitions of \( {\xx_{i+1}} \) and \( {\gg_i} \) into this expression, we get:
\begin{align}
	F({\xx_{i+1}}) &\leq F({\xx_i}) - \frac{1}{L} (\nabla F({\xx_i}) + {\ee_i})^\top \nabla F({\xx_i}) \notag \\
	&\quad + \frac{1}{2L} \|\nabla F({\xx_i}) + {\ee_i}\|^2 \notag \\
	&= F({\xx_i}) - \frac{1}{L} \|\nabla F({\xx_i})\|^2 - \frac{1}{L} \nabla F({\xx_i})^\top {\ee_i} \notag \\
	&\quad + \frac{1}{2L}\|\nabla F({\xx_i})\|^2 + \frac{1}{L} \nabla F({\xx_i})^\top {\ee_i} \notag \\
	&\quad + \frac{1}{2L}\|{\ee_i}\|^2 \notag \\
	&= F({\xx_i}) - \frac{1}{2L}\|\nabla F({\xx_i})\|^2 + \frac{1}{2L}\|{\ee_i}\|^2. \label{14}
\end{align}

We now use \eqref{13} to derive a lower bound for the norm of \( \nabla F({\xx_i}) \) in relation to the optimality of \( F({\xx_i}) \). By minimizing both sides of \eqref{13} with respect to \( y \), it follows that the minimum on the left-hand side is achieved when \( \yy = \xx^\ast \); similarly, the minimizer of the right-hand side is given by \( \yy = \xx - (1/\mu)\nabla f(\xx) \). Thus, we have:
\begin{align}
	F(\xx^\ast) &\geq F(\xx) - \frac{1}{\mu} \nabla F(\xx)^\top \nabla F(\xx) \notag \\
	&\quad + \frac{1}{2\mu} \|\nabla F(\xx)\|^2 = F(\xx) - \frac{1}{2\mu} \|\nabla F(\xx)\|^2.
\end{align}

By rearranging and specializing for \( \xx = {\xx_i} \), we obtain:
\begin{align}
	\|\nabla F({\xx_i})\|^2 \geq 2\mu (F({\xx_i}) - F(\xx^\ast)). \label{16}
\end{align}

Finally, subtract \( F(\xx^\ast) \) from both sides of \eqref{14} and apply \eqref{16} to derive the following:
\begin{align}
	F({\xx_{i+1}}) - F(\xx^\ast) &\leq F({\xx_i}) - F(\xx^\ast) - \frac{\mu}{L} (F({\xx_i}) - F(\xx^\ast)) \notag \\
	&\quad + \frac{1}{2L} \|{\ee_i}\|^2 \notag \\
	&= \left( 1 - \frac{\mu}{L} \right) [F({\xx_i}) - F(\xx^\ast)] + \frac{1}{2L} \|{\ee_i}\|^2\label{19}
\end{align}


Typically, edge devices cannot collect all available training samples, i.e., $M < M_{\text{total}}$.  Because not all data samples are used, the convergence speed of this method is sublinear \cite{friedlander2012hybrid}. However, as the collected sample size increases, the error in the computed gradient \( {\gg_k} \) decreases, allowing the sample size to be used as a mechanism to control gradient error.
Define \(\mathcal{B}= \{\mathcal{D}_k,\forall k \}\) as the set of all collected data, and define \( \mathcal{N} \) as the complement of \(\mathcal{B}\), such that \( \mathcal{B}\cup \mathcal{N}_k = \{1, \dots, M_{\text{total}}\} \). The residual of the gradient, given by \eqref{11b}, satisfies the following expression:
\begin{align}
	{\ee_i} = \frac{M_{\text{total}} - |\mathcal{B}|}{M_{\text{total}} |\mathcal{B}|} \sum_{i \in \mathcal{B}} \nabla f({\xx_i,\xi_i}) - \frac{1}{M_{\text{total}}} \sum_{i \in \mathcal{N}} \nabla f({\xx_i},\xi_i). 
\end{align}

The first term accounts for a weighted adjustment of the gradient estimate, while the second term captures the gradient component not included in the sampled set. Based on Assumption \ref{as3}, we can establish a bound on the norm of this residual relative to the full-sample gradient
\begin{align}
	&\| \ee_i \|^2\\
	&= \left\| \left( \frac{M_{\text{total}} - |\mathcal{B}|}{M_{\text{total}} |\mathcal{B}|} \right) \sum_{i \in \mathcal{B}} \nabla f(\xx_i,\xi_i) - \frac{1}{M_{\text{total}}} \sum_{i \in \mathcal{N}} \nabla f(\xx_i,\xi_i) \right\|^2 \notag \\
	&\leq \left( \frac{M_{\text{total}} - |\mathcal{B}|}{M_{\text{total}} |\mathcal{B}|} \right)^2 \left( \left\| \sum_{i \in \mathcal{B}} \nabla f(\xx_i,\xi_i) \right\| + \frac{1}{M_{\text{total}}} \left\| \sum_{i \in \mathcal{N}} \nabla f(\xx_i,\xi_i) \right\| \right)^2 \notag \\
	&\leq \left( \frac{M_{\text{total}} - |\mathcal{B}|}{M_{\text{total}} |\mathcal{B}|} \right)^2 \left( \sum_{i \in \mathcal{B}} \|\nabla f(\xx_i,\xi_i) \| + \frac{1}{M_{\text{total}}} \sum_{i \in \mathcal{N}} \|\nabla f(\xx_i,\xi_i) \| \right)^2 \notag \\
	&\leq 4 \left( \frac{M_{\text{total}} - |\mathcal{B}|}{M_{\text{total}}} \right)^2 \left( \beta_1 + \beta_2 \|\nabla F(\xx_i) \|^2 \right).
\end{align}
Next, in the same way that \eqref{16} is derived, we use Assumption \ref{as1} to derive the upper bound 
\begin{align}
	\|\nabla F({\xx_i})\|^2 \leq 2L(F({\xx_i}) - F(\xx^\ast)). \label{20a}
\end{align}
Thus the upper bound for $\|\ee_i\|^2$ can be expressed using the ratio of sample sizes and the gap to the optimality:
\begin{align}
	\|\ee_i\|^2 \leq 4 \left[ \frac{M_{\text{total}} - |\mathcal{B}|}{M_{\text{total}}} \right]^2 \left( \beta_1 + 2 \beta_2 L [F(\xx_k) - F(\xx^*)] \right).\label{error_2}
\end{align}
Combing \eqref{19} and \eqref{error_2}, we obtain
\begin{align}
	&F({\xx_{i+1}}) - F(\xx^\ast) \leq \left( 1 - \frac{\mu}{L} \right) [F({\xx_i}) - F(\xx^\ast)] \notag\\
	&\quad\quad+2\left[ \frac{M_{\text{total}} - |\mathcal{B}|}{M_{\text{total}}} \right]^2 \left( \beta_1/L + 2 \beta_2  [F(\xx_i) - F(\xx^*)] \right)\label{43}
\end{align}
Recursively applying \eqref{43} for $i+1$ times, we obtain 
\begin{align}\label{41}
	&F({\xx_{I}}) - F(\xx^\ast) \leq \Psi(\mathcal{B})^I[F(\xx_0)-F(\xx)]\notag\\&+2\left[ \frac{M_{\text{total}} - |\mathcal{B}|}{M_{\text{total}}} \right]^2\frac{\beta_1}{L}\frac{1-\Psi(\mathcal{B})^I}{1-\Psi(\mathcal{B})}
\end{align}
where $\Psi(\mathcal{B})=\left[(1-\frac{\mu}{L})+4\left(\frac{M_{\text{total}}-|\mathcal{B}|}{M_{\text{total}}}\right)^2\beta_2\right]$.

	\bibliographystyle{ieeetran}
	\bibliography{IEEEabrv,mybib}
\end{document}